\documentclass{elektr}
\usepackage{hyperref}
\hypersetup{
colorlinks=true,
urlcolor=blue,
citecolor=blue}
\usepackage[all]{xy,xypic}
\usepackage{amsfonts,amssymb,amsmath,amsgen,amsopn,amsbsy,theorem,graphicx,epsfig}
\usepackage{eufrak,amscd,bezier,latexsym,mathrsfs,eurosym,enumerate}
\usepackage[utf8]{inputenc}
\usepackage[english]{babel}
\usepackage{cleveref,multirow}
\usepackage[dvipsnames]{xcolor}
\usepackage[pagewise]{lineno}
\usepackage{caption}
\usepackage{subcaption}

\linenumbers

\newcommand{\argmax}{\operatornamewithlimits{argmax}}

\yil{}
\vol{}
\fpage{}
\lpage{}
\doi{}

\title{ Joint  optimization of TWT  mechanism  and scheduling for IEEE 802.11ax  }

\author[AUTHOR and AUTHOR]{
\textbf{Mehmet  KARACA$^{1}$\thanks{mehmet.karaca@tedu.edu.tr}}\\
Department of Electrical and Electronics Engineering, Engineering Faculty, TED University, Ankara, Turkey, \\ ORCID iD: https://orcid.org/0000-0002-2425-2013\\
\\ [1.8em]

\rec{.201}
\acc{.201}
\finv{..201}
}

\def\E{\ifmmode{\mathbb E}\else{$\mathbb E$}\fi} 
\def\N{\ifmmode{\mathbb N}\else{$\mathbb N$}\fi} 
\def\R{\ifmmode{\mathbb R}\else{$\mathbb R$}\fi} 
\def\Q{\ifmmode{\mathbb Q}\else{$\mathbb Q$}\fi} 
\def\C{\ifmmode{\mathbb C}\else{$\mathbb C$}\fi} 
\def\H{\ifmmode{\mathbb H}\else{$\mathbb H$}\fi} 
\def\Z{\ifmmode{\mathbb Z}\else{$\mathbb Z$}\fi} 
\def\P{\ifmmode{\mathbb P}\else{$\mathbb P$}\fi} 
\def\T{\ifmmode{\mathbb T}\else{$\mathbb T$}\fi} 
\def\SS{\ifmmode{\mathbb S}\else{$\mathbb S$}\fi} 
\def\DD{\ifmmode{\mathbb D}\else{$\mathbb D$}\fi} 

\newcommand{\bse}{\begin{subequations}}
\newcommand{\ese}{\end{subequations}}
\newcommand{\ben}{\begin{enumerate}}
\newcommand{\een}{\end{enumerate}}
\newcommand{\bens}{\begin{enumerate*}}
\newcommand{\eens}{\end{enumerate*}}
\newcommand{\be}{\begin{equation}}
\newcommand{\ee}{\end{equation}}
\newcommand{\bea}{\begin{eqnarray}}
\newcommand{\eea}{\end{eqnarray}}
\newcommand{\baa}{\begin{eqnarray*}}
\newcommand{\eaa}{\end{eqnarray*}}
\newcommand{\bc}{\begin{center}}
\newcommand{\ec}{\end{center}}

\newtheorem{theorem}{Theorem}

\theoremstyle{corollary}

\theoremstyle{lemma}
\newtheorem{lemma}{Lemma}

\theoremstyle{proposition}

\theoremstyle{axiom}

\theoremstyle{conjecture}

\theoremstyle{example}

\theoremstyle{definition}

\theoremstyle{remark}

\begin{document}
\nolinenumbers
\maketitle

\begin{abstract}
IEEE 802.11ax as the newest Wireless Local Area Networks (WLANS) standard brings enormous improvements in network throughput,  coverage and energy efficiency in densely populated areas. Unlike previous IEEE 802.11 standards where power saving mechanisms have a limited capability and flexibility,  802.11ax comes with a different mechanism called Target Wake Time (TWT) where stations (STAs) wake up only   after each  TWT interval and different STAs can wake up  at different time instance  depending on their application requirements. As an example, for a periodic data arrival occurring in  IoT applications, STA can wake up  by following the data period and go to sleep mode for a  much longer time, and STAs with high traffic volume can have shorter TWT interval to wake up more frequency. Moreover, as multi-user transmission capability is added to 802.11ax, multiple STAs can have the same TWT interval and wake up at the same time, and hence there is a great opportunity to have collision-free transmission by scheduling multiple STAs on appreciate TWT intervals to reduce energy consumption and also increase network throughput.  In this paper, we investigate the problem  of STAs scheduling and TWT interval assignment together to reduce overall energy consumption of the network. We propose an algorithm that dynamically selects STAs to be served and assigns them the most suitable TWT interval given their traffic and channel conditions.  We analyze  our algorithm through Lyapunov optimization framework and show that our algorithm is arbitrarily close to the optimal performance at the price of	increased queue sizes.  Simulation results show that our algorithm consumes less power and support higher traffic compared to a benchmark algorithm that operates randomly for TWT assignment. 

\end{abstract}

\section{Introduction}
\label{Int}
IEEE 802.11 based Wireless Local Area Networks (WLANS) has been the most popular way to reach wireless internet for their great flexibility and cost efficiency. Since the last decade, the increasing demand for data-hungry services and  the deployment of Wi-Fi hotspots almost everywhere has triggered the new standardization activities, and the first signal was the need for the development of IEEE 802.11ac~\cite{ref1}  that can provide gigabit rates with the enhancements made in mostly PHY layer. However, the improvements in PHY layer has been reaching  to its limits, and also the  dense network scenarios has become the  one of the main  challenges, as a consequence MAC layer improvements for 802.11 based WLANs has become inevitability. To this end, IEEE 802.11ax ( (High Efficiency WLAN (HEW))  task group~\cite{ref2} has started the new standardization activity, which aims at enhancing PHY and MAC layer in a radical way and is expected to end at 2020.  

The main target of 802.11ax is to increase MAC throughput per station at least four times compared to 802.11ac, while providing spectral efficiency and improved energy saving in densely populated networks.  The  MAC structure of the previous 802.11 WLANs is mostly based on  Carrier Sense Multiple Access with Collision Avoidance  (CSMA/CA() mechanism that is simple but  inefficient especially when the network is dense and the traffic is high. As a result, new MAC structure has started to be developed and Orthogonal frequency division multiple access (OFDMA) has been selected as one of the new access schemes for 802.11ax to improve the performance of dense networks by offering multi-user transmission capability. Also, Multi-user MIMO (MU-MIMO) that enables simultaneous multiple transmission at the same time and frequency resource can be realized for uplink transmission with 802.11ax~\cite{ref3}.

Moreover, in order to improve the energy efficiency of STAs, the IEEE
802.11ax  utilizes a new power saving mechanism which does not exist with the previous WLANs standards but was already introduced in 802.11ah~\cite{ref4, ref5} which  is called  Target
Wake Time (TWT) mechanism.   With TWT mechanism each STA  and Access Point (AP) agree on a common set of parameters that include when STAs wake up (TWT interval) and how long they stays awake, which allows STAs to receive/transmit data at only  the wake-up times, and in the rest of the time  they can stay in sleep mode. There are two types  of TWT
mechanism defined in the standard. With \textit{broadcast} TWT mechanism, the TWT parameters are carried on beacon frames and STAs  need to wake up to receive the beacon in order to know that their TWT information.  

With \textit{individual} TWT,   STAs can have different wake-up period (i.e.,TWT interval) depending on their traffic conditions~\cite{ref6}, which  is  one of the key differences from the previous power saving mechanisms  that STAs can have only a common beacon period as their wake-up interval.  Also, the  difference from 802.11ah is that as multiuser transmission is possible with 802.11ax, multiple STAs can share the same wake-up period and simpler  scheduling scheme among STAs can be developed to allocate TWT intervals to STAs and hence the contention between STAs can be reduced significantly with individual TWT mechanism~\cite{ref6}. We note that the assignment of wake-up period to STAs and the scheduling  are up to the implementation and  not specified in the standard. Another feature of TWT mechanism is that the agreement  can be done for uplink, downlink or both of these transmissions, and also a STA whose TWT parameters have been already agreed can wake-up at any time instance by using the Distributed Coordination Function (DCF). In this paper, we focus on the joint TWT interval assignment and scheduling  to further reduce energy consumption with a careful consideration  of traffic condition of each STA.

Since the early announcement of the new IEEE 802.11 standard,  the attention from both industry and academia to 802.11ax has increased, and the work~\cite{ref7,ref8} explain and investigate the new features coming with 802.11ax but the interest in these studies is not the power saving mechanism.  In~\cite{ref9,ref10}, the authors are interested in reducing power consumption of 802.11ax based networks but by not directly attacking the TWT mechanism. Specifically, the work in~\cite{ref9} tries to save energy by developing a scheduling scheme  for the case when OFDMA is used. Also, a cloud-based centralized management is proposed for the energy saving of a dense network with many 802.11ax based APs in~\cite{ref10}. The "clock drift" problem where the misalignment of the scheduling times of STAs within TWT mechanism occurs and its effect on the power consumption is investigated in~\cite{ref11}. The most relevant studies to our work are found in~\cite{ref6,ref12,ref13}. A detailed explanation and performance assessment  are given in~\cite{ref6}, and it also highlights that a collision-free transmission can be achieved by a proper employment of  a TDMA-type scheduling within TWT mechanism which is left as a open problem. In~\cite{ref12}, the optimization of broadcast TWT mechanism is performed for energy saving and the authors in~\cite{ref13} aims at reducing energy for uplink  multi-user transmission.  The works in~\cite{ref12,ref13} do not attempt to optimize TWT interval of STAs (e.g., wake-up period) and also do not consider a scenario where different STA may have different traffic conditions.

Our contributions are summarized as follows:
\begin{itemize}
	\item We aim at reducing the energy consumption of 802.11ax network where individual TWT is employed as the power saving mechanism. We improve the system model of the individual TWT  by giving more flexibility for the assignment of  TWT interval.

	\item We formulate the problem of joint  multi-user scheduling and TWT interval assignment for energy saving as a stochastic optimization problem by taking into account the traffic condition of each STA. Also, we  propose an algorithm that solves our problem by assigning TWT intervals to each STA dynamically depending on their traffic and channel conditions. We  find upper bounds on the optimal energy consumption and queue  sizes of STAs via Lyapunov optimization framework.

	\item Through simulations, we  demonstrate  the performance of our algorithm by comparing it with a benchmark algorithm that randomly assigns TWT intervals to STAs. We also depict the tradeoff between the average energy consumption and the average queue delay.
\end{itemize}

\section{System Model and Problem Definition}
We consider a fully-connected 802.11ax based WLAN network  where there is an AP serving $M$ users.  We assume random
channel gains between the AP and the STAs that are independent across time and STAs. As there is only a discrete finite set of $F$ Modulation and coding Schemes (MCS)
available in practice, only a fixed set of data rates denoted as $\mathcal{R} = {r_1, r_2, . . . , r_F }$
can be supported.

Recall that with the new sleeping and power saving mechanism introduced in 802.11ax,  each STA and AP agree a common set of TWT parameters (individual TWT), and STAs can only wake up at the specified time to minimize their energy consumption. Next, we explain two  important TWT parameters that are most relevant to this work and also represented in Figure~\ref{fig:Fig1}:

\begin{itemize}
	\item TWT session: the duration of time (in seconds) over that a STA wakes up and receive or transmit data. It  is assumed to be same for all STAs, and denoted as $t_{up}$.
	
	\item TWT interval: the  time interval between two consecutive wake-up time of a STA.
	
	More specifically, STAs wake up periodically to  be served by the AP. We consider individual TWT with which  different STAs can have different periodicity depending on their traffic and application requirements. Also,  unlike previous studies, we consider that this periodicity is not static but can be changed dynamically on the agreement of STAs and the AP. Specifically, at the beginning of  a \textit{epoch} that starts at $t =nT$ where $n = 1,2,....$, \textit{TWT Wake-up interval}  of a STA $m$ denoted as $T_m^t$ is determined. Then, the STA $m$ wakes up at every $T_m^t$  until the next epoch. Each 	epoch lasts  $T$ seconds,  where $T > T_m^t$ for all $t$ and $m$. At the beginning of the next epoch, $T_m^t$ is updated if necessary. 
	
	For instance, as shown in Figure~\ref{fig:Fig1}, the TWT wake-up interval for  STA $m$ at the beginning of epoch $n$ (i.e., at time $t=nT$)  is $T_m^t$. However, at the next epoch, $t+1=(n+1)T$ it is set to $T_m^{t+1}  > T_m^t $, and STA $m$ has less number of TWT sessions and more time to sleep. To depict the integration of the multi-user capability in TWT, we show in Figure~\ref{fig:Fig1} that STA $k$ and  STA $l$ have the same TWT interval at time $t$ and wake up at the same time.
\end{itemize} 
\begin{figure}[h]
	\centering
	\includegraphics[width=1\columnwidth]{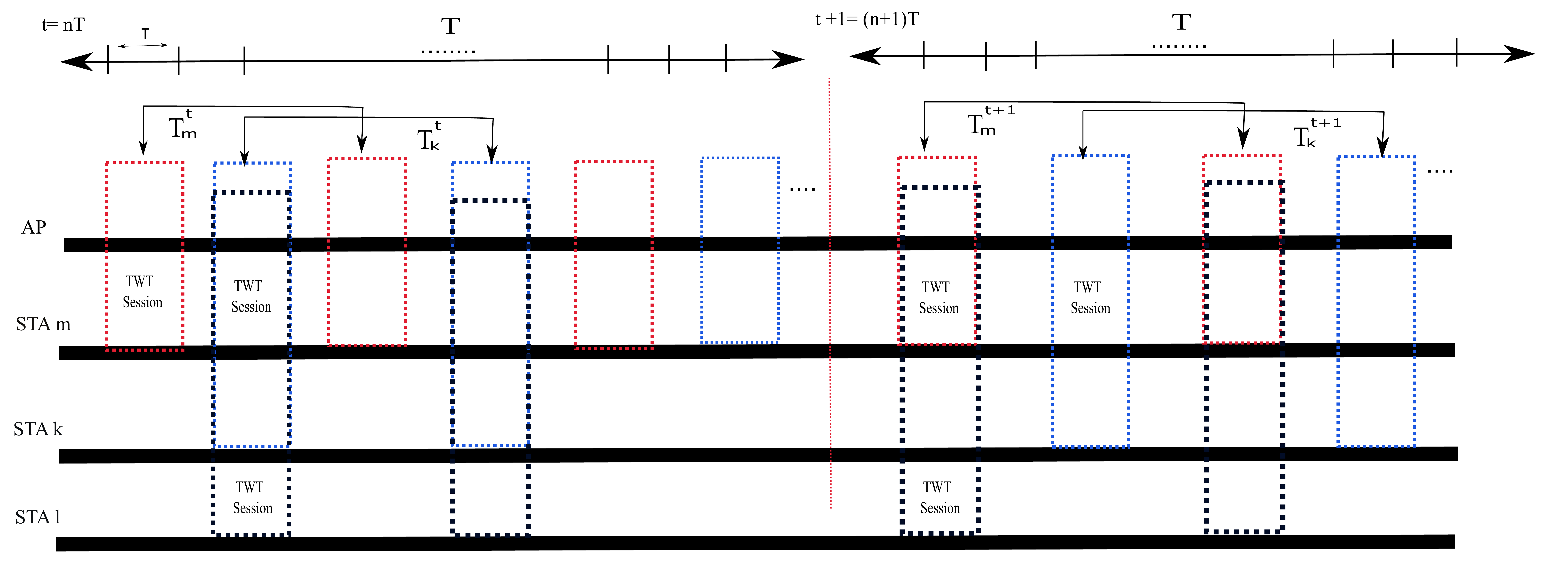}
	\caption{TWT with dynamic wake up time. }
	\label{fig:Fig1}
\end{figure}
The TWT wake-up interval can take values from a predefined discrete set, $T_m^t \in \{T_1, T_2, ..., T_L\}$ for all $m$ and $t$.  For a given $T_m^t$, the number of TWT sessions that STA $m$ can utilize at time $t=nT$ (over epoch $n$) is given as:
\begin{align*}
N_m^t = \left \lfloor \frac{T}{T_m^t} \right \rfloor,
\end{align*}
where $ \left \lfloor x \right \rfloor$ is the largest integer that is smaller than $x$. Since $T_m^t$ takes only discrete values from a finite set, $T$ as a system parameter and the set $\{T_1, T_2, ..., T_L\}$ can be chosen such that  $N_m^t =  \frac{T}{T_m^t} $\footnote{ Another parameter of TWT mechanism is the first time for a STA to start TWT session. Here, for analytical simplicity we assume TWT interval can fully represent the number of TWT sessions.}.  We further assume that each $T$ seconds is divided into mini-slots denoted as  $\tau$, $\tau\in \{0,1,2,\ldots\}$, and $\tau$ is multiple integer of all possible values of $T_m^t$. For simplicity, we assume that $\tau$ is equal to $t_{up}$.

We now determine the energy consumption for STA $m$ at a mini time slot $\tau$. Note that when STA $m$ wakes up, it becomes awake for a duration of $t_{up}$ seconds. $t_d$ and $t_u$ are the fractions of $t_{up}$ that are used for downlink and uplink transmission, respectively. Also, the power consumption during downlink and  uplink are denoted as $P_d$ and $P_u$, respectively. Then, the total energy consumption during a single TWT session is given by,
\begin{align*}
 E_s= P_d \times t_dt_{up} + P_u \times t_ut_{up}.
\end{align*}
Note that when $\tau = i T_m^t$, where $i>0$ is an integer,  STA $m$ wakes up to start a TWT session. When $\tau \neq i T_m^t$,  STA $m$ is at sleep mode. The energy consumption of STA $m$ at a mini slot $\tau$ is given as,
\begin{align}
E_m(\tau) =& \left\{ \begin{array}{l l}
E_s             & \text{; if  $\tau=iT_m^t$}\\
E_{sleep}     & \text{; otherwise},
\end{array} \label{eq:r2}
\right.
\end{align}
where $E_{sleep}$ is the energy consumed at sleep mode during of  $\tau$ seconds and $E_{sleep} = P_{sleep} \times \tau $, where  $P_{sleep}$ is the power consumption at sleep mode.

During the epoch $n$ starting at $t=nT$, there are $N_m^t$ number of TWT sessions for STA $m$. Thus, the STA $m$ consumes $E_m^t$ amount of energy over the epoch $n$, where 
\begin{align}
\sum_{\tau=1}^{T/\tau} E_m(\tau) = E_m^t= E_s \times N_m^t + \left(\frac{T}{\tau}-N_m^t\right)E_{sleep}, \label{eq:Emt}
\end{align}
where $\left(\frac{T}{\tau}-N_m^t\right)E_{sleep}$ accounts for the total energy consumption of STA $m$ at sleep mode since there are $\left(\frac{T}{\tau}-N_m^t\right)$ mini-slots for sleeping over the epoch.
Then, the time average energy consumption over all STAs is determined as follows:
\begin{align}
E_{avg} = \limsup_{t\rightarrow \infty}\frac{1}{t}
\sum_{t=0}^{t-1} \sum_{m=1}^{M}\mathbb{E} [E_m^t] \label{eq:power}
\end{align}

At each mini-slot $\tau$, data arrives to the queue of each users. Let $A_m(\tau)$ be the amount of data (i.e., in bits) arriving into
the queue of STA $m$ at time slot $\tau$. We assume that $A_m(\tau)$ is a
stationary process and it is independent across STAs and time
slots. We denote the arrival rate vector as $\boldsymbol
\lambda=(\lambda_1,\lambda_2,\cdots,\lambda_M)$, where $\lambda_m =
{\mathbb E}[A_m(\tau)]$. Let $\boldsymbol
Q(\tau)=(Q_1(\tau),Q_2(\tau),\cdots,Q_M(\tau))$ denote the vector of queue
sizes, where $Q_m(\tau)$ is  the queue length of user $m$ at time slot
$\tau$. A queue is strongly stable if 
\begin{equation}
\limsup_{S\rightarrow \infty}\frac{1}{S} \sum_{\tau=0}^{S-1}\mathbb{E}(Q_m(\tau)) < \infty. 
\end{equation}
Moreover, if every queue in the network is stable then the network
is called stable. The dynamics of the queue of STA $m$ is
\begin{equation}
Q_m(\tau)=\max[Q_m(\tau)-R_{m}(\tau) , 0] + A_m(\tau)  \label{eq:queuelength}
\end{equation}
We note that the transmission rate of $R_m(\tau)$ can take positive value, i.e., $R_m(\tau) > 0$ at time instances $\tau = iT_m^t$ as STA $m$  wakes up only at these instances. When $\tau \neq  iT_m^t$,  then $R_m(\tau) = 0$. $\Lambda$ denotes the capacity region of the network, which is the largest possible set of rates $\boldsymbol \lambda$ that can be supported by a joint scheduling and TWT interval assignment algorithm with ensuring the network stability, $\boldsymbol \lambda \in \Lambda$. 

As the multi-user transmission capability of 802.11ax is enhanced with the introduction of OFDMA and MU-MIMO in uplink, more than one STA can be assigned to a TWT interval. We assume that at most $K$ STAs can share the same TWT interval, and let $N_{T_l^t}$  be the number of STAs that are assigned to TWT interval $T_l$ at time $t$, $l \in \{1,2,...,L\}$. Therefore, in total up to $L \times K$ STAs can have a TWT interval assigned over each epoch. With our setup, the control decision of the network for STA $m$ is the determination of $T_m^t$. That is to say  at the beginning of each time $t$ (i.e., epoch), the network determines on the TWT wake-up interval for STA $m$. Then,  we consider the following optimization problem:
\begin{align}
&\min \quad E_{avg} \\
\text{s.t.:1)}\  &\text{Netwrok stability}\\
\text{2)}\ & T_m^t \in \{T_1, T_2, ..., T_L\} \quad \forall t, m\\
\text{3)}\ & N_{T_l^t} \leq K \quad \forall l \in \{1,2,...,L\}, \forall t
\label{eq:problem}
\end{align}
The constrain in (8) ensures that TWT interval decided for a STA takes values from a predefined discrete set. Also, the third constraint in (9) guarantees that at most $K$ STAs can be assigned to a TWT interval. At the beginning of each epoch, the system determines TWT wake-up interval for $L \times K$ number of STAs and other STAs can go to sleep mode during whole $T$ seconds to save  energy. 

The problem (6)-(7)-(8)-(9) aims at minimizing the total average energy consumption of STAs by properly choosing TWT wake-up interval with the consideration of their traffic conditions. Intuitively, it is expected that the solution of the problem must strike a balance between the energy consumption and queue size. In other words, when the channel condition is good and queue size is large it can reduce the TWT interval, and when  queue size is low it prefers a higher TWT interval for more energy saving. The problem constitutes a stochastic optimization problem and we next propose a solution based on Lyapunov optimization technique~\cite{ref14}.

\section{Joint TWT Interval and Scheduling Algorithm (JTWSA)}
In this work, we employ Lyapunov drift and optimization tools~\cite{ref14} that enables us to  deal with performance optimization (i.e., energy minimization in our case) and network stability problems together in a unified framework.  In order to use this tool, first, a Lyapunov function is  determined, and it is common in literature to choose a square function of queue sizes as the Lyapunov function:
	\begin{align*}
     L(\mathbf{Q}(t))\triangleq \frac{1}{2} \sum_{m=1}^M Q_m^2(t)
	\end{align*}
Lyapunov function  acts as a measurement tool for the total queue sizes in the network. Another important concept in this framework is Lyapnov drift that measures how much the expected values of the queue sizes vary within a one single slot. Here, we use  the conditional T-slot Lyapunov drift that is the expected variation in the Lyapunov function over $T$ slots and  given as:
\begin{equation}
\Delta_T(t)\triangleq \mathbb{E}
\left\{L(\mathbf{Q}(t+T))-L(\mathbf{Q}(t))|\mathbf{Q}(t) \right\}  \label{eq:drift}.
\end{equation}
To stabilize the network,  we need to make sure that the drift does not grow indefinitely and hence ultimate aim is to minimize the drift to obtain a stable network. For a more tractable analysis, an upper bound for the drift is determined and then this bound is minimized. 

In order to handle both the network stability and an objective optimization together, drift-plus penalty method~\cite{ref14} is employed. With this method,  our penalty (or cost) $V\sum_{m=1}^{M}\mathbb{E} [E_m^t | (\mathbf{Q}(t)]$ (i.e., $T$ slot drift-plus-penalty) is added to the drift~\eqref{eq:drift}
\begin{equation}
\Delta_T(t) + V\sum_{m=1}^{M}\mathbb{E} [E_m^t |  \mathbf{Q}(t)]  \label{eq:driftpluspen},
\end{equation}
where $V$ is a control parameter between the queue delay and the average system penalty. Now, in the Lyapunov optimization framework, we turn our energy minimization and network stability problem to be the minimization of \eqref{eq:driftpluspen}. Instead of directly attacking to minimize~\eqref{eq:driftpluspen},  in the following Lemma we first derive an upper bound on \eqref{eq:driftpluspen} and then we propose our algorithm that minimizes the bound. 
\begin{lemma}
	\label{lemma:1}
	Given $V >1$, and at time $t=nT$,  for any feasible control decision $(T_1^t,T_2^t,...,T_M^t $), we have the following bound,
	\begin{align}
	\Delta_T(t)  + V\sum_{m=1}^{M}\mathbb{E} [E_m^t |  \mathbf{Q}(t)]
	&\leq   B_1  + \mathbb{E} \left[\sum_{\tau = t}^{t +T-1} \sum_{m=1}^{M} Q_m(\tau)A_m(\tau)   | \mathbf{Q}(t)\right] \nonumber\\
	& - \mathbb{E} \left[\sum_{\tau = t}^{t +T-1} \sum_{m=1}^{M} Q_m(\tau) R_m(\tau) - V\sum_{m=1}^{M} E_m^t   | \mathbf{Q}(t) \right]\label{eq:bound2}
	\end{align}
	where $B_1= \frac{ MT (R_{max}^2 + A_{max}^2)}{2}$.
\end{lemma}
\begin{proof}
	\label{sec:lemma1}
	The proof starts with finding an upper bound for the Lyapunov drift given in \eqref{eq:drift} by using the following fact: for user $m$, the following inequality holds.
	\begin{align}
	Q_m^2(t+1) -   Q_m^2(t)  \leq   R_m^2(t) + A_m^2(t) + 2Q_m(t) \left [A_m(t) -  R_m(t)\right]
	\end{align}
	By summing (13) over   $[t,t+T-1]$ and knowing that $R_m(t) \leq R_{max}$ and $A_m(t) \leq A_{max}$  for all $m$ and $t$  we obtain,
	\begin{align}
	Q_m^2(t+T) -   Q_m^2(t)   \leq   TR^2_{max} + TA^2_{max} 
	 - 2 \left [ \sum_{\tau = t}^{t +T-1}Q_m(\tau)[R_m(\tau) - A_m(\tau)]\right]
	\end{align}
	Then,  by taking  the conditional expectation of (14) with respect to  $\mathbf{Q}(t)$ and summing over all users, and dividing by  1/2 we have,
	\begin{align}
	\Delta_T(t)  \leq   B_1  - \mathbb{E} \left[\sum_{\tau = t}^{t +T-1} \sum_{m=1}^{M} Q_m(\tau)[R_m(\tau) - A_m(\tau) ]   | \mathbf{Q}(t)\right] 
	\end{align}
	Finally, we add the penalty term $V\sum_{m=1}^{M}\mathbb{E} [E_m^t |  \mathbf{Q}(t)]$ to both sides of (14) and rearranging the resulting terms, we have the bound in~\eqref{eq:bound2}. This completes the proof.
\end{proof}
The Lyapunov optimization framework now helps to  solve our primary problem by  finding an algorithm that minimizes the right hand side (R.H.S) of the bound~\eqref{eq:bound2}.  Clearly, the solution of following problem minimizes the R.H.S of~\eqref{eq:bound2}.

\textbf{Problem 1:}
\begin{equation*}
\max_{T_1^t.T_2^t,...,T_M^t}  \mathbb{E} \left[\sum_{\tau = t}^{t +T-1} \sum_{m=1}^{M} Q_m(\tau)[R_m(\tau) -A_m(\tau)]-  V\sum_{m=1}^{M} E_m^t  | \mathbf{Q}(t) \right]\label{eq:Alg1}
\end{equation*}

However, the solution of Problem 1 requires a prior knowledge of queue size and transmission rates of STAs, i.e., future information. Particularly, obtaining the future traffic information is not easy so we follow the same idea proposed in~\cite{ref15} to find a looser but more relaxed bound without needing the future queue size information. With this idea, we approximate  the future queue sizes as the current observation, i.e., $Q_m(\tau) = Q_n(t)$ for all $\tau \in [t, t+T-1]$ for all STA $m$, and now Problem 1.1. aims at maximizing the following problem:

\textbf{Problem 1.1:}
\begin{equation*}
\max_{T_1^t.T_2^t,...,T_M^t}  \mathbb{E} \left[\sum_{\tau = t}^{t +T-1} \sum_{m=1}^{M} Q_m(t)R_m(\tau) -  V\sum_{m=1}^{M} E_m^t  | \mathbf{Q(t)} \right]\label{eq:Alg2}
\end{equation*}

Now, the solution of Problem 1.1  needs only queue size information at the beginning of $t=nT$ but it sill needs future transmission rates over $T$ seconds. Fortunately, for WLANs, the coherence time of the channel is long enough since STAs are usually stationary and the mobility is low. Also, transmission rates can be estimated with the advanced machine learning techniques~\cite{ref16}. We assume that the transmission rates change every $T$ seconds, hence the Problem 1.1 is reduced to the following deterministic optimization problem as follow:

\textbf{Problem 1.2:}
\begin{align}
\max_{T_1^t.T_2^t,...,T_M^t}   \sum_{m=1}^{M} Q_m(t)R_m(t)N_m^t -  V E_m^t  \label{eq:Alg31}
\end{align}
By using~\eqref{eq:Emt}, we can rewrite~\eqref{eq:Alg31} as,

\begin{align}
\max_{T_1^t.T_2^t,...,T_M^t}   \sum_{m=1}^{M} Q_m(t)R_m(t)N_m^t -  V N_m^t (E_s - E_{sleep}) \label{eq:Alg32}
\end{align}

The following Joint TWT Assignment and Scheduling Algorithm (JTWSA) solves~\eqref{eq:Alg32}.

\textbf{ JTWSA :}
\begin{itemize}
	\item Input: $V$, $P_d$, $P_u$, $P_{sleep}$, $t_{up}$. At every $t=nT$ do:
	\item Step 1: Among all users, find 
		\begin{align*}
		\mathcal{S}(t) = \argmax-(L\times K) \left\{  Q_m(t)R_m(t) - V(E_s-E_{sleep}) \right\}, 
		\end{align*}
		where $ \argmax-(L\times K) $ is  the operation of choosing the first $L \times K$ elements of a set of real numbers sorted in decreasing order.

	\item Step 2:  
	
	- Start assigning TWT intervals $T_1,T_2,..,T_L$ to the STAs in $\mathcal{S}(t)$  if only $Q_m(t)R_m(t) > V(E_s-E_{sleep})$. For instance, the TWT interval of the first $K$ STAs in  $\mathcal{S}(t)$ is set to $T_1$, for the second group of $K$ STAs in  $\mathcal{S}(t)$, it is set to $T_2$ so on. If the number of STAs is assigned to TWT interval is $K$, go to the next TWT interval.
	
	 - If $Q_m(t)R_m(t) \leq V(E_s-E_{sleep})$, set the TWT interval for that STA $m$ to $T$.

	\item Step 3: After determining $T_m^t$ for every STA,  perform transmit/receive operation and updates queue until next $t=(n+1)T$. 
	       
\end{itemize}

As JTWSA optimally minimizes the R.H.S of (12),  we have the following Theorem that shows performance of JTWSA.
\begin{theorem}(Lyapunov Optimization)
	\label{thm:1}
	Suppose $\boldsymbol \lambda$ is an interior point in $\Lambda$, then there
	exits a constant $\epsilon > 0$ such that $\boldsymbol \lambda + \boldsymbol \epsilon \in
	\Lambda$. Then, under JTWSA, we have
	the following bounds:
	\begin{align}
	\limsup_{T'\rightarrow \infty} \frac{1}{T'} \sum_{t=0}^{T'-1}
	\mathbb{E} [ E_m^t ] \leq E_{avg}^* +  \frac{B_2}{V}
	\label{eq:thr1}
	\end{align}
	and
	\begin{align}
	\limsup_{N\rightarrow \infty} \frac{1}{N} \sum_{n=0}^{N-1}
	\sum_{m=1}^{M} \mathbb{E} [Q_m(nT) ]\leq  \frac{B_2 + VE_{max}}{\epsilon}
	\label{eq:thr1}
	\end{align}
	Where $E_{avg}^*$ is the optimal solution of problem (6)-(7)-(8)-(9) and where $B_2= \frac{ MT^2 (R_{max}^2 + A_{max}^2)}{2}$ and $E_{max} = MTE_s $.
\end{theorem}
\begin{proof}
	To avoid redundancy with existing literature, we omit the details here and refer the readers to the proof in Theorem 5.4 of~\cite{ref14} and Theorem 2 in~\cite{ref15}. The sketch of the proof is as follows: first, it shows the existence of  a stationary randomized algorithm   that is optimal and achieves the minimum time average energy by choosing $T_m^t$ independent from $\mathbf{Q}(t)$ but according to a fixed probability distribution known to the network. Then, it is shown that JTWSA is better than this randomized algorithm in minimizing the R.H.S of (12) and thus it is also optimal. 
\end{proof}
Theorem 1 implies that the average energy under JTWSA approaches to  the minimum optimal energy $E_{avg}^*$ as $V$  increase, while the average  queue length  also increase with higher values of $V$. Next, we verify our findings with simulation results.

\section{Numerical Results}
We consider an IEEE 802.11 ax based WLAN where there is an AP serving $M =50$
STAs. We assume that all TWT sessions are used for uplink transmissions and choose the  parameters given in the following table.

\begin{table}[h]
\begin{center}
	\begin{tabular}{ | l |  p{5cm} |}
		\hline
	\textbf{	Simulation Parameters} & \textbf{Values}   \\ \hline
		T & 1 second, 2 seconds  
		\\ \hline
		$P_u$ & 1 W \\ \hline
		$P_{sleep}$ & 0.15 W \\ \hline
		$t_{up}$ & 1 ms \\ \hline
		$t_{d}$ & 0 \\ \hline
	    $t_{u}$ & 1 \\ \hline
		$R$ & [10, 20, 50, 100, 150, 200] Mbps \\ \hline
		$M$ & 50 STAs \\ \hline
		$K$ & 5 STAs\\ \hline
	    $V$ & [1000,  5000]\\ \hline
	    $\lambda$ & Possion traffic. Arrival period [5 4.5 4 3.5 3 2.5 2 1.5 1 0.5] seconds \\ \hline
	    File size &  25 KBytes\\ \hline
	    TWT values & Total 9 non-overlapping  TWT intervals. Minimum and maximum TWT intervals: 50 ms and 450 ms\\ 
		\hline
	\end{tabular}
\caption{Simulation parameters and values}
\end{center}
\end{table}
The network operates as follows: at the beginning of each $T$ seconds, the queue sizes and transmission rates for each STA is determined, and then the network assigns TWT intervals to STAs according to JTWSA. At most $K=5$ STAs can be assigned to same TWT interval thanks to the multi-user transmission capability of 802.11ax. As there are  $L=9$ different TWT intervals and $K=5$, at most 45 STAs can wake up for transmission over the epoch and other 5 STAs go to sleep mode during entire $T$ seconds. After assigning TWT intervals, STAs wake up at the determined intervals and transmit their data to AP. We compare JTWSA with a benchmark algorithm that assigns TWT intervals to STAs randomly at the beginning of every $T$ seconds. Next, we show the performance of JTWSA and the benchmark algorithm in terms of average queue sizes and energy consumption.

\begin{figure}
	\begin{subfigure}[b]{0.52\textwidth}
		\includegraphics[width=\textwidth]{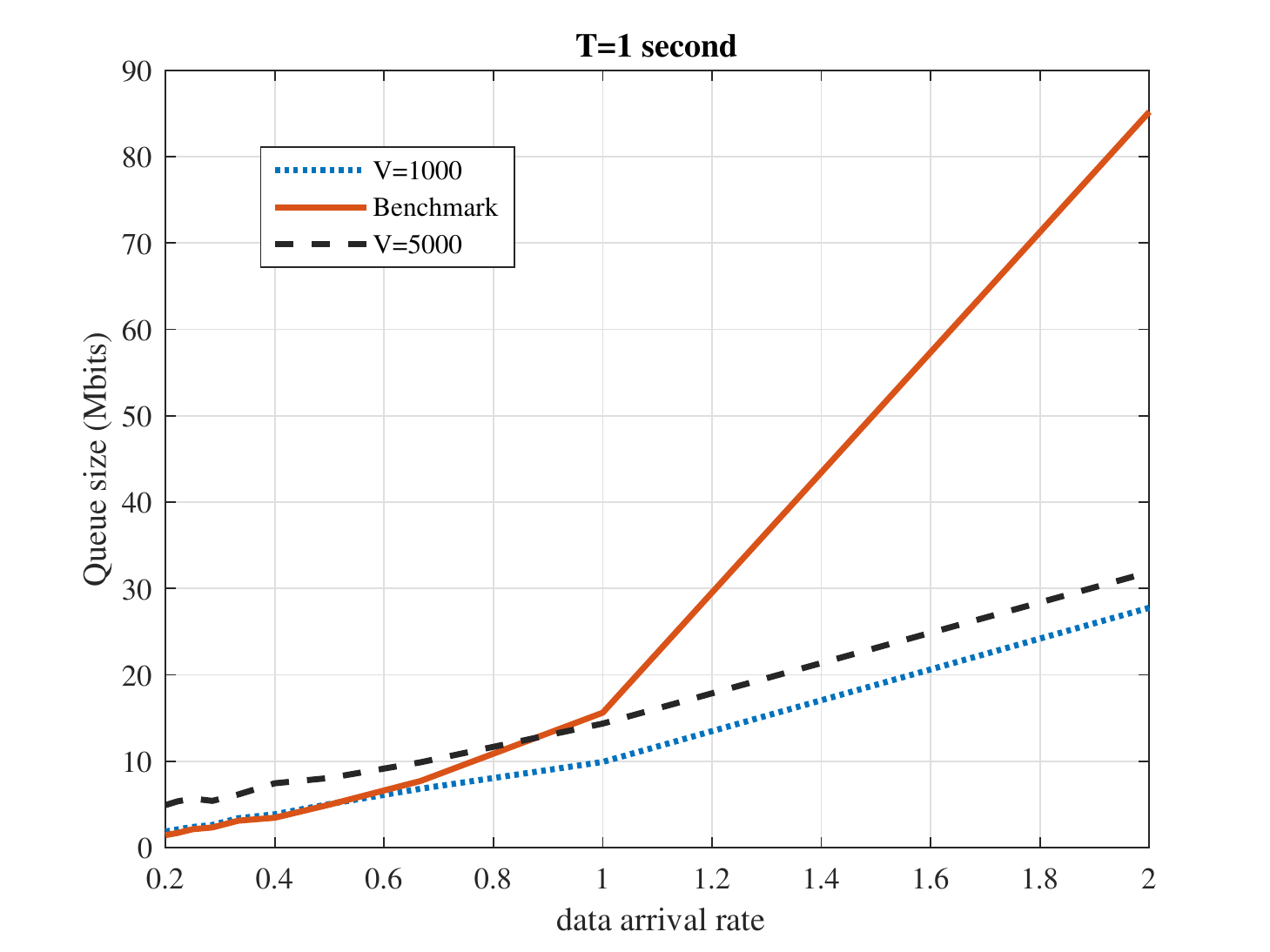}
		\caption{Avg. queue size with different  $\lambda$}
		\label{fig:Fig2}
	\end{subfigure}
	\begin{subfigure}[b]{0.52\textwidth}
		\includegraphics[width=\textwidth]{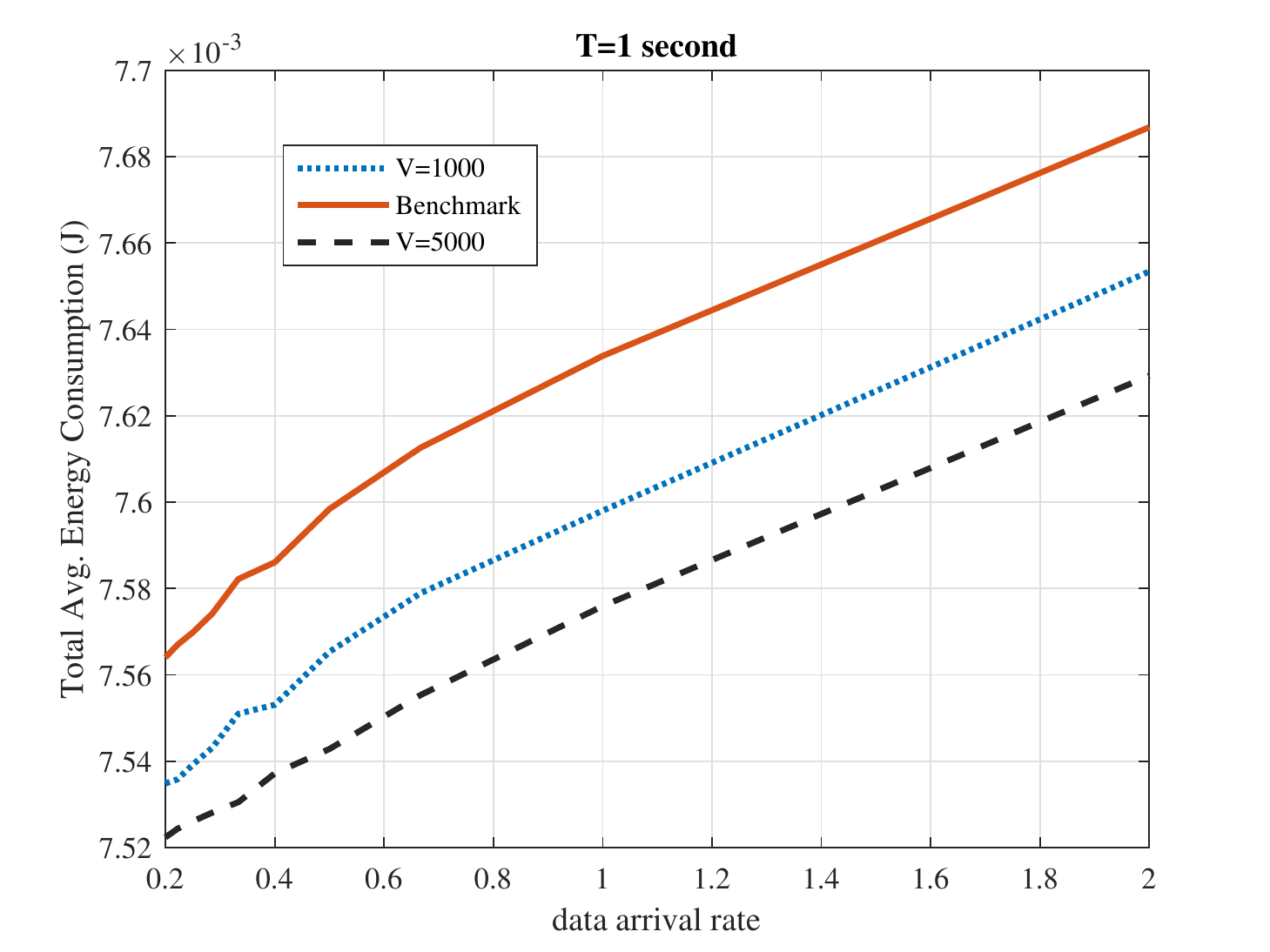}
		\caption{Avg. energy consumption with $\lambda$}
		\label{fig:Fig3}
	\end{subfigure}
\caption{Queue sized and energy consumption when T=1}\label{fig:2}
\end{figure}

Figure~\ref{fig:Fig2} and~\ref{fig:Fig3}  depict the total average queue size and energy consumption under JTWSA and the benchmark algorithm when $T$ is 1 second, respectively. From~\ref{fig:Fig2}, we see that as the arrival rate increases the total average queue size increases as well. However, there is a sudden increase in the queue size with the benchmark algorithm when the rate is around 1, which indicates that the network is not stable any more and the benchmark cannot support the arrival rate higher than 1. On the other hand, JTWSA can keep stabilizing  the network up to the arrival rate of 2. Recall that  Theorem 1 points that when the network operate with higher $V$ values,  a larger queue size is resulted in. Figure~\ref{fig:Fig2}  verifies that  the queue size with $V=5000$ is higher than the one occurred when $V=1000$.

Figure~\ref{fig:Fig3} shows the total average energy consumption of  JTWSA and the benchmark algorithm. Note that STAs can wake up even their queue size and data rate are low with the benchmark algorithm and hence the energy conservation is expected to be lower. This intuition is  demonstrated in Figure~\ref{fig:Fig3} where the energy consumption with the benchmark is higher than JTWSA. Also, as Theorem 1 suggests that as we increase $V$, we get closer to the global minimum energy consumption. Clearly, from Figure~\ref{fig:Fig3}, a higher $V$ value yields lower energy consumption for  JTWSA. Therefore, by choosing the control parameter $V$, a trade-off between queue delay and energy consumption can be achieved.
\begin{figure}
	\begin{subfigure}[b]{0.52\textwidth}
		\includegraphics[width=\textwidth]{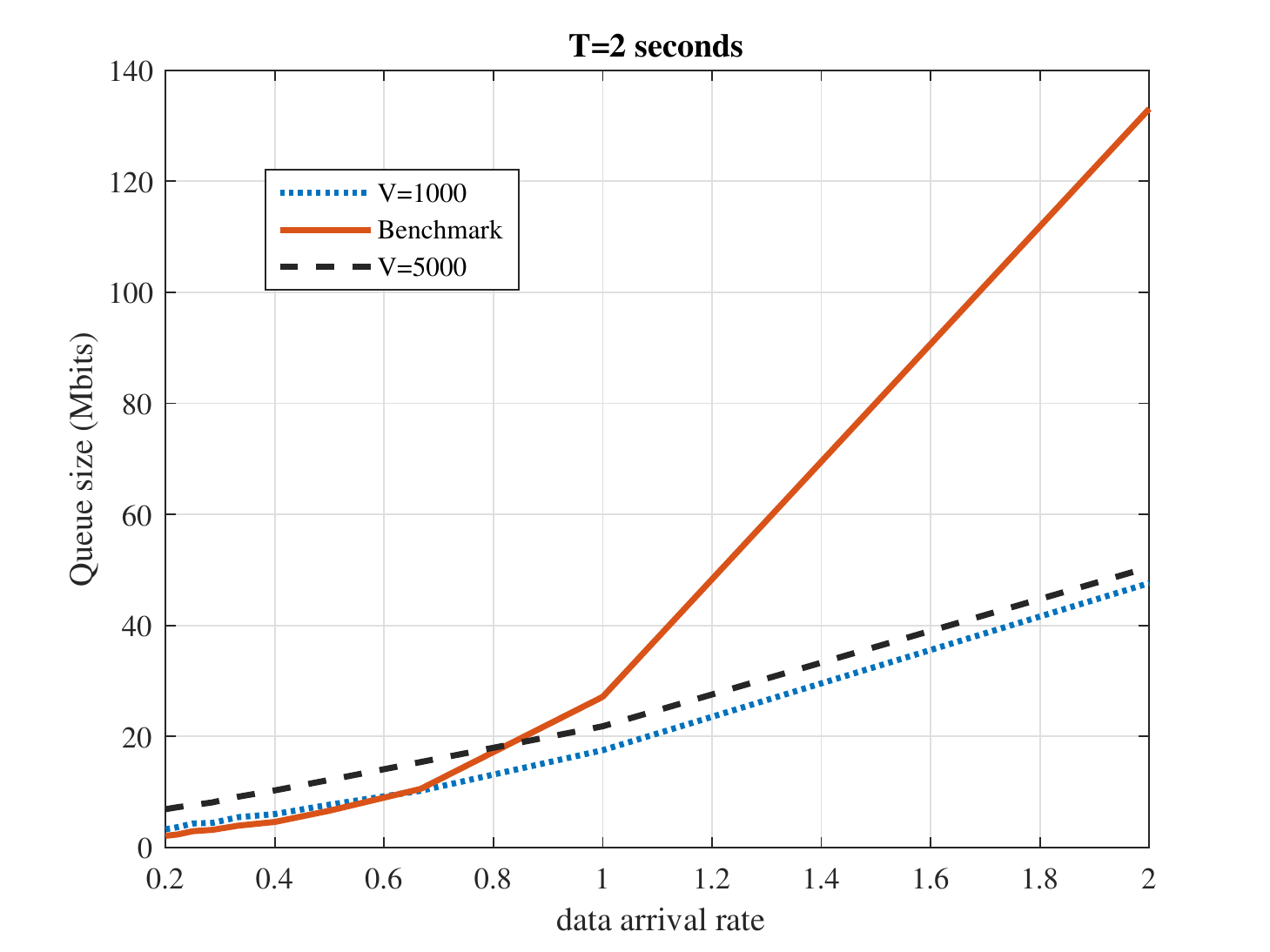}
		
		\caption{Avg. queue size with different  $\lambda$}
		\label{fig:Fig4}
	\end{subfigure}
	\begin{subfigure}[b]{0.52\textwidth}
		
		\includegraphics[width=\textwidth]{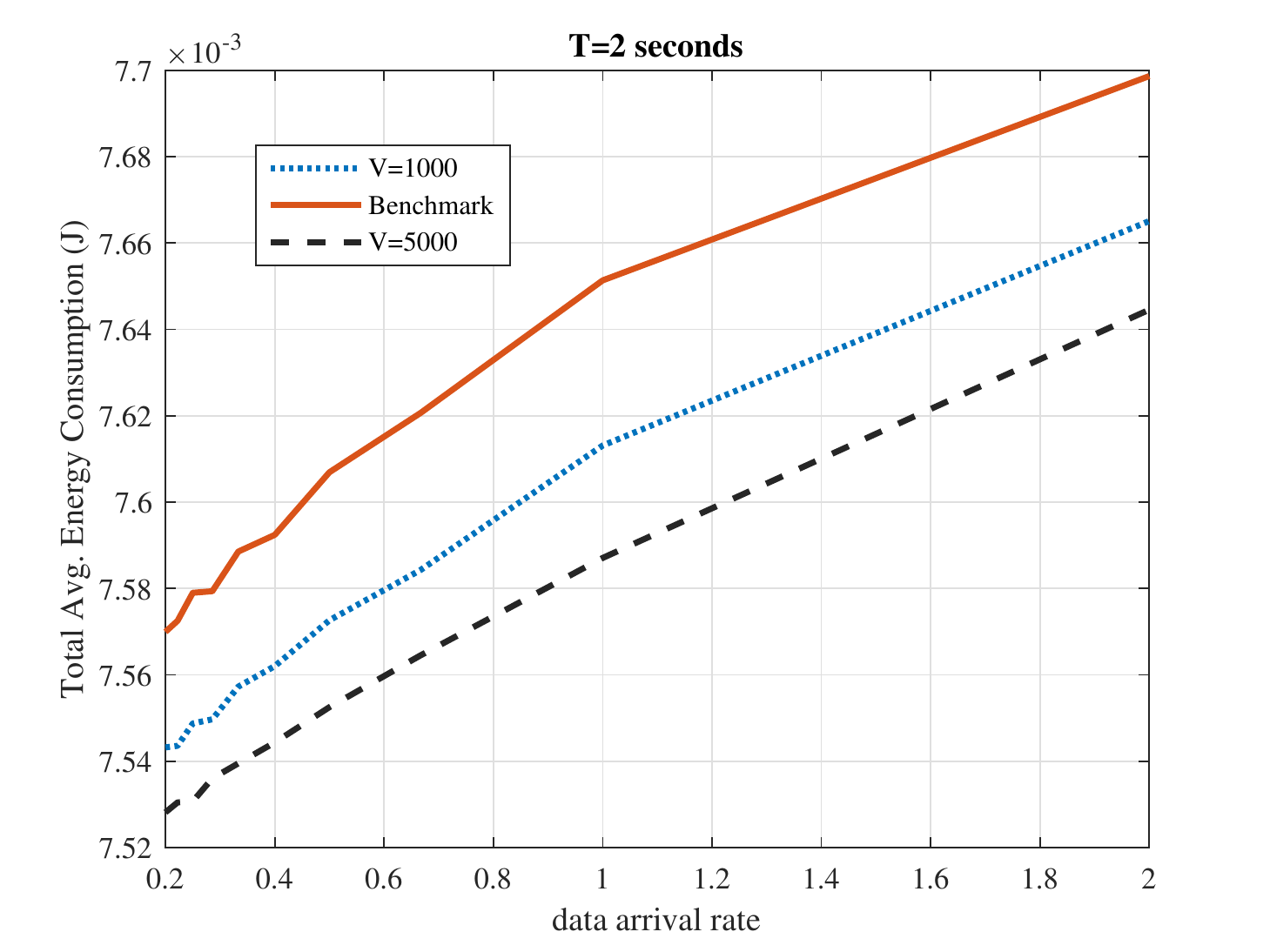}
		\caption{Avg. energy consumption with $\lambda$}
		\label{fig:Fig5}
	\end{subfigure}
	\caption{Queue sized and energy consumption when T=2}\label{fig:3}
\end{figure}

Lastly, Figure~\ref{fig:Fig4} and~\ref{fig:Fig5} depict the impact of a higher value of $T$ on the queue size and energy consumption. Specifically, we set $T=2$ seconds and according Theorem 1, both the queue size and energy consumption increase as we increase $T$. By comparing Figure~\ref{fig:Fig2} and Figure~\ref{fig:Fig4}, one can conclude that the queue size is larger with higher values of $T$. Similar result can be deduced for the energy consumption from Figure~\ref{fig:Fig3} and~\ref{fig:Fig5}. In order to close the minimum energy consumption when $T$ is high, we also need a higher $V$ values. However, the impact of a higher $T$ value on the energy consumption  is less significant but the impact on the queue size cannot be ignored.  We note although choosing lower $T$ values seem to be better in terms of queue size and energy, it will increase signaling overhead since the AP must communicate and agree on TWT interval at every $T$ second. This trade-off can be resolved given the operator preferences for a target objective and performance.

\section{Conclusion}
In this paper, we have developed an algorithm that aims at improving the energy efficiency of the power saving mechanism of IEEE 802.11ax  by optimization TWT intervals of STAs according to their traffic and channel conditions. Our algorithm assigns lower TWT interval when queue size is large  and/or transmission rate is high, and allows more sleeping time if traffic is low enough.  We have analyzed the algorithm via Lyapunov optimization framework and found performance bounds on the energy consumption and queue sizes. We have defined two important trade-off: first,  high energy saving can be achieved when a large average queue delay is allowed. Also, lower average delay can be obtained if TWT intervals are updated more frequently but the cost in this case is high signaling overhead. The impact of inaccurate transmission rates prediction on the performance of our algorithm is worth to be investigated. Also, in a multi-WLAN scenario where interference is significant, determining TWT intervals and channel allocation would be an interesting research problem.

\end{document}